\documentclass[letterpaper, 10 pt, conference]{ieeeconf}  

\IEEEoverridecommandlockouts                              
\usepackage{graphicx}
\usepackage{amsmath, amssymb}
\usepackage{mathtools}
\usepackage[T1]{fontenc}
\usepackage{lmodern}
\usepackage{newtxtext, newtxmath}
\usepackage{subcaption}
\usepackage[noend]{algpseudocode}
\usepackage{algorithm}
\usepackage{url}
\usepackage{cite}
\allowdisplaybreaks[3]

\newcommand{\T}{{\top}}

\newcommand{\pldev}[2]{{\partial {#1} / \partial {#2}}}
\newcommand{\bm}[1]{{\boldsymbol{\mathrm{#1}}}}

\newtheorem{theorem}{Theorem}

\newtheorem{definition}{Definition}

\newtheorem{lemma}{Lemma}

\algrenewcommand\algorithmicrequire{\textbf{Input:}}
\algrenewcommand\algorithmicensure{\textbf{Output:}}
\algrenewcommand\algorithmicforall{\textbf{for each}}

\title{\LARGE \bf
On First Integrals of Hamiltonian System with Holonomic Hamiltonian
}

\author{Tomoyuki Iori
\thanks{This work was partly supported by JSPS KAKENHI Grant Number JP21K21285.}
\thanks{T. Iori is with the Department of Information and Physical Sciences, Graduate School of Information Science and Technology, Osaka University, 1-5 Yamadaoka, Suita, Osaka 565--0871, Japan {\tt\small t-iori@ist.osaka-u.ac.jp}}%
\thanks{This work has been submitted to the IEEE for possible publication. Copyright may be transferred without notice, after which this version may no longer be accessible.}%
}

\begin{document}

\maketitle
\thispagestyle{empty}
\pagestyle{empty}

\begin{abstract}
	In this study, the solution of the Hamilton-Jacobi equation (HJE) with holonomic Hamiltonian is investigated in terms of the first integrals of the corresponding Hamiltonian system. 
	Holonomic functions are related to a specific type of partial differential equations called Pfaffian systems, whose solution space can be regarded as a finite-dimensional real vector space. 
	In the finite-dimensional solution space, the existence of first integrals that define a solution of the HJE is characterized by a finite number of algebraic equations for finite-dimensional vectors, which can be easily solved and verified. 
	The derived characterization was illustrated through a numerical example. 
\end{abstract}

\section{Introduction \label{sec:intro}}
The Hamilton-Jacobi equation (HJE) is one of the most fundamental equations in the analysis and control of nonlinear systems. 
It appears in various problems, such as optimal control~\cite{Bryson1975,Anderson1979}, $H_\infty$ control~\cite{Schaft1992}, and balanced realization~\cite{Scherpen1993} of nonlinear systems. 
For linear time-invariant systems, the HJE is reduced to the algebraic Ricatti equation, a set of algebraic equations for an unknown matrix, which can be solved using various solution methods. 
However, for nonlinear systems, the HJE is formulated as a nonlinear partial differential equation (PDE), which is difficult to solve analytically. 
Numerical solution methods for the HJE have been studied based on various mathematical techniques, such as series expansion~\cite{Lukes1969}, expansion with basis functions~\cite{Beard1997}, and data-driven approximation~\cite{Ito2020}. 

Hao et al.~\cite{Hao2020} proposed a numerical solution method for the HJE by generating functions that consider the initial and terminal boundary conditions of optimal control problems. 
The Taylor-series expansion of a generating function up to any prescribed order is computed by solving first-order differential equations. 
The approximated generating function is used to find a family of optimal control for different boundary conditions. 
However, in these numerical approaches, the computational complexity rapidly increases with the state dimension, known as the curse of dimensionality. 

The differential-geometric approach is an effective strategy for analyzing the HJE~\cite{Sakamoto2002}.
From the perspective of symplectic geometry, the solutions of the HJE are identified with Lagrangian manifolds in the cotangent bundle of the state space. 
A Lagrangian manifold is characterized by a finite number of first integrals of the corresponding Hamiltonian system as a set of their common zeros. 
Hence, solving the HJE is reformulated as the problem of finding the appropriate first integrals of the corresponding Hamiltonian system. 
This characterization reduces the HJE to a set of ordinary differential equations called the Lagrange-Charpit system, which is still difficult to solve analytically. 

Moreover, algebraic approaches based on commutative algebra are promising techniques, especially for avoiding the curse of dimensionality. 
For time-invariant Hamiltonian, Ohtsuka~\cite{Ohtsuka2011} characterized the solutions of the HJE with algebraic gradients in terms of the existence of involutive zero-dimensional ideals in a polynomial ring; 
Kawano and Ohtsuka~\cite{Kawano2013a} extended these results to the case of time-varying Hamiltonian. 
It should be noted that the first integrals that define a Lagrangian manifold are closely related to involutive zero-dimensional ideals; 
Indeed, the existence of first integrals in the polynomial ring implies the existence of an involutive zero-dimensional ideal. 

Considering these differential-geometric and algebraic approaches, in this study, the HJE was analyzed from the perspective of the theory of $D$-modules, which is a field of mathematics that studies the algebraic structures of PDEs and their solution space~\cite{Coutinho1995}. 
In the theory of $D$-modules, the solution set of PDEs is characterized by an ideal in the noncommutative ring of differential operators, which can be computed using symbolic computation. 
The symbolic computation of differential operators has been intensively studied~\cite{Saito2000,Oaku2003,Kauers2011} and its applications can be found in statistics~\cite{Nakayama2011,Kume2018} and moment problems~\cite{Brehard2019}. 
Moreover, for specific ideals, called zero-dimensional ideals, the solutions are called holonomic functions and constitute a finite-dimensional real vector space. 
This finite dimensional property is promising for finding the first integrals because any first integral in the solution space can be determined by a finite number of parameters. 

In this paper, the first integrals of the Hamiltonian system are investigated through the symbolic computation of differential operators. 
First, it is assumed that all first integrals, including the Hamiltonian, that define a Lagrangian manifold are included in the solution space of a zero-dimensional ideal. 
This assumption allows us to associate the first integrals with vector-valued functions that satisfy a specific type of first-order PDEs, called a Pfaffian system. 
Moreover, any solution of the Pfaffian system is characterized by only the boundary value at a single point where the Pfaffian system is defined. 
This leads to the characterization of first integrals by a finite number of algebraic equations for the boundary value. 
Through this characterization, the first integrals can be determined by solving algebraic equations, if such first integrals exist. 

\paragraph*{Notations}
For the field of real numbers $\bm{R}$ and a vector of indeterminates $y = [y_1\ \cdots\ y_n]^\T$, $\bm{R}(y)$ denotes the field of rational functions in $y_1,\dots,y_n$ over $\bm{R}$. 
$\partial_y \coloneqq [\partial_{y_1}\ \cdots\ \partial_{y_n}]^\T$ denotes a vector of differential operators, where $\partial_{y_i} = \pldev{}{y_i}$. 
Here, $\partial_y$ and $\partial_{y_i}$ are abbreviated by $\partial$ and $\partial_i$, respectively, if $y$ is clearly specified according to the context. 
For a multi-index $\alpha = (\alpha_1, \dots, \alpha_n) \in \bm{Z}^n_{\ge 0}$, $y^\alpha$ and $\partial^\alpha$ denote the monomial $y_1^{\alpha_1} \cdots y_n^{\alpha_n}$ and differential operator $\partial_1^{\alpha_1} \cdots \partial_n^{\alpha_n}$, respectively. 
Let $\bm{R}[y]$ and $\bm{R}(y)$ be the set of all polynomials and rational functions in $y_1,\dots,y_n$, respectively. 
The set of all $(m \times l)$ matrices and $(m \times m)$ skew matrices with components in $\bm{R}(y)$ is denoted by $\bm{R}(y)^{m \times l}$ and $\mathrm{Skew}_m(\bm{R}(y))$, respectively. 
The symbol $\bm{R}(y)\langle \partial \rangle$ denotes the noncommutative ring of differential operators with coefficients in $\bm{R}(y)$. 
The action of a differential operator $\mathcal{P}$ on a sufficiently smooth function $F(y) = F(y_1,\dots, y_n)$ is denoted by $\mathcal{P} \bullet F(y)$; for instance, $\partial_i \bullet F(y) = \pldev{F}{y_i}(y)$. 
If $\mathcal{P} \bullet F = 0$, a differential operator $\mathcal{P}$ is said to \emph{annihilate} function $F$ and $F$ is a \emph{solution} of $\mathcal{P}$. 
The left ideal in $\bm{R}(y)\langle \partial \rangle$ generated by a finite number of differential operators $\mathcal{P}_1,\dots,\mathcal{P}_s$ is defined as $\langle \mathcal{P}_1, \dots, \mathcal{P}_s \rangle \coloneqq \{\mathcal{Q}_1\mathcal{P}_1 + \cdots + \mathcal{Q}_s\mathcal{P}_s \mid \mathcal{Q}_1,\dots,\mathcal{Q}_s \in \bm{R}(y)\langle \partial \rangle\}$. 
In this paper, the adjective ``left'' is omitted and it is referred to as an ``ideal'' because there is no right ideal.
If $\mathcal{P} \bullet F = 0$ for all $\mathcal{P} \in I$, the ideal $I \subset \bm{R}(y)\langle \partial \rangle$ is said to \emph{annihilate} function $F$ and $F$ is a \emph{solution} of $I$. 

\section{First Integrals of Hamiltonian Systems and Hamilton-Jacobi Equation\label{sec:problem}}
For a scalar-valued function $h(x, p)$, we consider the following first-order PDE for scalar-valued function $v(x)$: 
\begin{equation}
		h(x, p) = 0, \quad p = \nabla_x v(x), \label{eq:HJE}
\end{equation}
where $x \in \bm{R}^n$. 
The PDE~\eqref{eq:HJE} is called the HJE for the Hamiltonian $h$. 
The Hamiltonian system associated with $h$ is a dynamical system with $2n$ state variables $z = (x, p)$ defined as follows. 
\begin{equation}
	\begin{cases}
		\dot{x} &= \nabla_p h(x, p) \\
		\dot{p} &= -\nabla_x h(x, p). 
	\end{cases} \label{eq:Hsys}
\end{equation}
A first integral of the Hamiltonian system~\eqref{eq:Hsys} is a function that is constant along the trajectory of~\eqref{eq:Hsys}, i.e., a function $f$ that satisfies
\begin{equation*}
	\dot{f} = (\nabla_x f)^\T \dot{x} + (\nabla_p f)^\T \dot{p} = \{h, f\} = 0, 
\end{equation*}
where $\{\cdot, \cdot\}$ denotes the Poisson bracket, defined as
\[
\{f, g\} = (\nabla_p f)^\T \nabla_x g - (\nabla_x f)^\T \nabla_p g. 
\]

The first integrals of the Hamiltonian system are closely related to the solutions of the HJE. 
Let $\bar{z} = (\bar{x}, \bar{p})$ be a fixed point with $h(\bar{z}) = 0$ and $\pi$ be the projection to the first component, i.e., $\pi\colon z = (x, p) \mapsto x $. 
For a neighborhood $U \subset \bm{R}^{2n}$ of $\bar{z}$, a smooth function $v(x)$ defined on $\pi(U)$ is a solution of the HJE if and only if $h(z) = 0$ holds in the following subset
\[
\Lambda_v \coloneqq \{z = (x, p) \in U \mid p = \nabla_x v(x)\} \subset \bm{R}^{2n}. 
\]
By symplectic geometry, $\Lambda_V$ can be written as
\[
\Lambda_v = \{z \in U \mid f_1(z) = f_2(z) = \cdots = f_n(z) = 0 \}, 
\]
where $f_1 = h$ and $f_2,\dots,f_n$ are the first integrals of~\eqref{eq:Hsys} such that the following conditions hold (see \cite{Sakamoto2002} for details): 
\begin{align}
&\{ f_k(z), f_l(z) \} = 0 \quad (z \in U,\ 1 \leq k < l \leq n), \label{eq:poisson}\\
&\det[\nabla_p f_1(\bar{z})\ \cdots\ \nabla_p f_n(\bar{z})] \neq 0. \label{eq:projective}
\end{align}

Let us consider the problem to find analytic functions $f_2,\dots,f_n$ at $\bar{z}$ that satisfy the system of nonlinear PDEs~\eqref{eq:poisson} as well as the inequality condition~\eqref{eq:projective}, which is difficult to solve analytically. 
The set of all analytic functions at $\bar{z}$ comprises an infinite dimensional vector space; this is reflected in the series expansion approach and Galerkin method because they usually require an infinite number of parameters to exactly determine a solution of the HJE, which cannot be accomplished in finite time.

To overcome this difficulty, the notion of Pfaffian system associated with holonomic functions is used. 
The finite-dimensional solution space of a Pfaffian system allows us to exactly reduce the problem of finding analytic functions $f_2,\dots,f_n$ to that of determining finite dimensional vectors, which can be solved easily without approximations. 

\section{Holonomic Functions and Pfaffian Systems}
In this section, some notions related to the symbolic computation of differential operators are introduced, referring to~\cite{Saito2000,Hibi2013,Kauers2011,Kauers2013} for most of the definitions and lemmas. 

First, let us consider the notion of holonomic functions, where all first integrals are assumed to be included. 
Holonomic functions can be defined as the solution of zero-dimensional ideals. 
\begin{definition}[\cite{Hibi2013}] \label{def:zero}
	For an ideal $I \subset \bm{R}(y)\langle \partial \rangle$, let $\bm{R}(y)\langle \partial \rangle / I$ be the quotient space as $\bm{R}(y)$-linear spaces.  
	The ideal $I$ is said to be \textit{zero-dimensional} if $\bm{R}(y)\langle \partial \rangle / I$ is finite-dimensional. 
\end{definition}
\begin{definition}[\cite{Hibi2013}] \label{def:HF}
	A function $f(y)$ is said to be \textit{holonomic} if it is a solution of a zero-dimensional ideal in $\bm{R}(y)\langle \partial \rangle$. 
\end{definition}

In this paper, $h(x, p)$ is called the \textit{holonomic Hamiltonian} if it is holonomic as a function of $2n$ variables $z = (x, p)$. 
The following lemma can be used to verify whether a given function is holonomic. 
\begin{lemma}
	A function $f(y)$ with $y = [y_1\ \cdots\ y_m]^\T$ is holonomic if it is a solution of differential operators $\mathcal{P}_i$: 
	\begin{equation}
	\mathcal{P}_i \bullet f(y) = 0 \quad (i=1,\dots,m),
	\end{equation}
	where each $\mathcal{P}_i$ is defined as a finite sum of the following form: 
	\begin{equation}
	\mathcal{P}_i = \sum_{k} c_{ik}(y) \partial_i^k \quad \left( c_{ik} \in \bm{R}(y) \right). \label{eq:defHF}
	\end{equation}
\end{lemma}
\begin{proof}
	The proof is readily completed because the ideal $\langle \mathcal{P}_1,\dots,\mathcal{P}_m\rangle$ is zero-dimensional (see Theorem~6.1.11 in~\cite{Hibi2013}). 
\end{proof}

The class of holonomic functions is closed under addition, multiplication, and differentiation. 
\begin{lemma}[\cite{Kauers2013}] \label{lem:closure}
	Suppose $f(y)$ and $g(y)$ are holonomic functions. 
	Then, $f(y) + g(y)$, $f(y)g(y)$, and $\partial_i f(y)\ (i=1,\dots,m)$ are also holonomic. 
\end{lemma}

The zero-dimensional ideals that annihilate the sums and products of holonomic functions can be computed through the symbolic computation of differential operator; however, it requires the notion of holonomic ideals in the Weyl algebra (see~\cite{Oaku2003} for details).

A specific system of PDEs, called the Pfaffian system, can be derived from the finite dimensional quotient space $\bm{R}(y)\langle \partial \rangle / I$. 
\begin{definition}[\cite{Saito2000}] \label{def:Pfaff}
	The Pfaffian system defined by $A_1,\dots,A_m \allowbreak \in \bm{R}(y)^{s \times s}$ is a system of PDEs for an $s$-dimensional vector-valued function $q(y)$: 
		\begin{equation}
			\partial_i q(y) = A_i(y)q(y), \quad (i=1,\dots,m). \label{eq:Pfaff}
		\end{equation}
\end{definition}

The relationship between holonomic functions, zero-dimensional ideals, and Pfaffian systems can be summarized as follows. 
\begin{lemma}[\cite{Saito2000}] \label{lem:HFandPfaff}
	Let $I \subset \bm{R}(y)\langle \partial \rangle$ be a zero-dimensional ideal such that the quotient space $\bm{R}(y)\langle \partial \rangle / I$ is $s$-dimensional. 
	Let $\mathcal{B}$ be a vector of differential operators defined as 
	\begin{equation}
		\mathcal{B} \coloneqq [1\ \partial^{\alpha_1}\ \cdots\ \partial^{\alpha_{s-1}}]^\T, \label{eq:ODvec}
	\end{equation}
	where $1, \partial^{\alpha_1}, \dots, \partial_{s-1}$ are the standard monomials~\cite{Saito2000} of $I$. 
	Then, $A_1(y),\dots,A_m(y) \in \bm{R}(y)^{s \times s}$ in~\eqref{eq:Pfaff} can be computed such that a vector-valued function $Q(y) = \mathcal{B} \bullet f(y)$ satisfies the Pfaffian system, where $f$ denotes any solution of $I$. 
	Moreover, $A_1(y),\dots,A_m(y)$ satisfy the integrability condition: 
	\begin{equation}
		\partial_j A_i + A_i A_j = \partial_i A_j + A_j A_i \quad (1 \leq i < j \leq m)\label{eq:intCond}
	\end{equation}
\end{lemma}

The key property of the Pfaffian system is that its solution space is finite-dimensional over $\bm{R}$ as the space of analytic functions. 
\begin{lemma}[\cite{Saito2000}] \label{lem:finiteness}
	With the same notations as Lemma~\ref{lem:HFandPfaff}, let $U$ denote any simply connected domain in $\{y \in \bm{R}^m \mid D(y) \neq 0\}$, where $D(y)$ is the least common multiple of all denominators included in $A_1, \dots, A_m$. 
	Let $\mathbb{S}(I)$ be the set of all analytic solutions of $I$ defined on $U$. 
	Let $\mathbb{VS}(A_1, \dots, A_m)$ be the set of all analytic solutions of the Pfaffian system~\eqref{eq:Pfaff} defined on $U$. 
	Then, $\mathbb{S}(I)$, $\mathbb{VS}(A_1,\dots,A_m)$, and $\bm{R}^s$ are isomorphic to each other with isomorphisms $\phi_{\mathcal{B}}\colon \mathbb{S}(I) \to \mathbb{VS}(A_1,\dots,A_m)$ and $\psi_{\bar{y}}\colon \mathbb{VS}(A_1,\dots,A_m) \to \bm{R}^s$ defined by
	\begin{equation*}
		\phi_{\mathcal{B}}(y) \coloneqq \mathcal{B} \bullet f, \quad \psi_{\bar{y}}(q) \coloneqq q(\bar{y}),
	\end{equation*}
	with any fixed $\bar{y} \in U$. 
	The inverse $\phi_{\mathcal{B}}^{-1}(q)$ can be defined as the first component of $q$; $\psi_{\bar{y}}^{-1}(\bar{q})$ can be defined as the unique solution of~\eqref{eq:Pfaff} with boundary condition $q(\bar{y}) = \bar{q}$. 
\end{lemma}

When $I = \langle \mathcal{P}_1,\dots,\mathcal{P}_m \rangle$, $\mathcal{B}$ and $A_1,\dots,A_m$ can be computed from $\mathcal{P}_1,\dots,\mathcal{P}_m$ using the Gr\"{o}bner bases in $\bm{R}(y)\langle \partial \rangle$~\cite{Saito2000}. 
In this paper, the solution space $\mathbb{S}(I)$ is denoted by $\mathbb{S}(\mathcal{P}_1,\dots,\mathcal{P}_m)$ if $I = \langle \mathcal{P}_1,\dots,\mathcal{P}_m \rangle$. 

\section{First Integrals in Solution Space of Pfaffian System}
For differential operators $\mathcal{P}_1, \dots, \mathcal{P}_{2n}$ of the form of ~\eqref{eq:defHF} that annihilate $h(x, p)$, suppose that other first integrals $f_2,\dots,f_n$ that define a solution of the HJE are included in the solution space $\mathbb{S}(\mathcal{P}_1,\dots,\mathcal{P}_{2n})$. 
Now, a condition for the existence of such first integrals is derived by considering the Pfaffian system associated with $\mathcal{P}_1, \dots, \mathcal{P}_{2n}$. 
This condition is obtained as a set of algebraic equations for a finite number of finite-dimensional vectors, which can be easily solved. 
Consequently, the first integrals can be determined by solving the algebraic equations if they have a sufficient number of solutions. 

The condition can be obtained in two steps. 
First, the set of PDEs~\eqref{eq:poisson} is converted into an infinite number of algebraic equations for some finite-dimensional vectors. 
It should be noted that this first part is presented in the preliminary form in the conference proceedings~\cite{Iori2021}. 
Subsequently, the number of equations is reduced to a finite number.  
In this section, $\mathcal{B}$ and $A_1,\dots,A_{2n} \in \bm{R}(z)^{d \times d}$ denote the vector of differential operators and matrices of rational functions derived from $\mathcal{P}_1,\dots,\mathcal{P}_{2n}$ through Lemma~\ref{lem:HFandPfaff}. 
Moreover, $\mathbb{S}$ and $\mathbb{VS}$ denote $\mathbb{S}(\mathcal{P}_1,\dots,\mathcal{P}_{2n})$ and $\mathbb{VS}(A_1, \dots, A_{2n})$, respectively. 
In addition, a point $\bar{z} \in U$ is fixed, where the domain $U$ is defined as in Lemma~\ref{lem:HFandPfaff}. 

\subsection{Reduction from PDEs to Infinite Number of Algebraic Equations}
For the first step, suppose that all first integrals $f_1, \dots, f_n$ lie in $\mathbb{S}$. 
From Lemma~\ref{lem:HFandPfaff}, there exists $q_k \in \mathbb{VS}$ such that $q_k = \phi_{\mathcal{B}}(f_k)$ for each $f_k$. 

From~\eqref{eq:ODvec}, we obtain
\begin{equation}
	\begin{cases}
		\nabla_x f_k(z) = B_x(z)q_k(z), \\
		\nabla_p f_k(z) = B_p(z)q_k(z),
	\end{cases} \label{eq:derivs}
\end{equation}
where $B_x \in \bm{R}(z)^{n \times d}$ and $B_p \in \bm{R}(z)^{n \times d}$ consist of the first rows of $A_1,\dots,A_n$ and those of $A_{n+1},\dots,A_{2n}$, respectively. 
By substituting~\eqref{eq:derivs} into~\eqref{eq:poisson} and~\eqref{eq:projective}, we obtain
\begin{align}
	&\begin{aligned}
	&\{f_k, f_l\}(z) = (q_k^\T \Omega q_l)(z) = 0, \\
	&\qquad \qquad \qquad (z \in U,\ 1 \leq k < l \leq n)\label{eq:Poisson2Quad} 
	\end{aligned} \\
	&\det \left\{B_p(\bar{z})\left[q_1(\bar{z})\ \cdots\ q_n(\bar{z})\right]\right\} \neq 0 \label{eq:projCond}
\end{align}
for a point $\bar{z} \in \mathbb{U}$ and $\Omega \coloneqq B_p^\T B_x - B_x^\T B_p \in \mathrm{Skew}_d(\bm{R}(z))$. 

From Lemma~\ref{lem:finiteness}, $q_k^\T \Omega q_l$ is analytic at any $\bar{z} \in U$; hence, $q_k^\T \Omega q_l = 0$ over $U$ if and only if 
\begin{equation}
	\left. \partial^\alpha \bullet \left( q_k^\T \Omega q_l \right)(z)\right|_{z = \bar{z}} = 0 \label{eq:allderivzero}
\end{equation}
holds for any $\alpha \in \bm{Z}_{\geq 0}^{2n}$. 
For each differential operator $\partial_i\ (i=1,\dots,2n)$, the Pfaffian system yields
\begin{align*}
	\partial_i \bullet \left( q_k^\T \Omega q_l \right) &= (\partial_i q_k)^\T \Omega q_l + q_k^\T (\partial_i \Omega) q_l + q_k^\T \Omega (\partial_i q_l) \\
	&= q_k^\T \left( A_i^\T \Omega + \partial_i \Omega + \Omega A_i \right) q_l. 
\end{align*}
By defining a mapping $\mathcal{D}_i\colon \mathrm{Skew}_d(\bm{R}(z)) \to \mathrm{Skew}_d(\bm{R}(z))$ for each $i = 1,\dots,2n$ as
\begin{equation}
	\mathcal{D}_i \Omega \coloneqq A_i^\T \Omega + \Omega A_i + \partial_i \Omega, \label{eq:defD}
\end{equation}
we obtain
\begin{equation}
	\partial^\alpha \bullet (q_k^\T \Omega q_l) = q_k^\T (\mathcal{D}^\alpha\Omega) q_l. \label{eq:pd2D}
\end{equation}

By substituting~\eqref{eq:pd2D} into~\eqref{eq:allderivzero}, an infinite set of algebraic equations for vectors $\bar{q}_k$ and $\bar{q}_l$ is obtained: 
\begin{equation}
	\bar{q}_k^\T \left\{\left( \mathcal{D}^\alpha \Omega \right)(\bar{z})\right\} \bar{q}_l = 0 \quad (\alpha \in \bm{Z}_{\geq 0}^{2n}), \label{eq:infCond}
\end{equation}
where $\mathcal{D}^\alpha$ is the composition of mappings $\mathcal{D}_1^{\alpha_1} \circ \mathcal{D}_2^{\alpha_2} \circ \cdots \circ \mathcal{D}_{2n}^{\alpha_{2n}}$. 
Thus far, the discussion can be summarized as follows.
\begin{theorem} \label{thm:infCond}
	Suppose the differential operators $\mathcal{P}_1, \dots, \allowbreak\mathcal{P}_{2n}$ are of the form of~\eqref{eq:defHF} and annihilate the holonomic Hamiltonian $h$. 
	For a fixed $\bar{z} \in U$ and $\bar{q}_1 = \psi_{\bar{z}} \circ \phi_{\mathcal{B}}(h)$, let $\bar{q}_2, \dots, \bar{q}_n \in \bm{R}^d$ be the constant vectors that satisfy~\eqref{eq:infCond} for $k, l \in \{1,\dots,n\}$. 
	Then, the functions $f_k \coloneqq \phi_{\mathcal{B}}^{-1} \circ \psi_{\bar{z}}^{-1}(\bar{q}_k) \in \mathbb{S}\ (k=2,\dots,n)$ satisfy conditions~\eqref{eq:poisson} and~\eqref{eq:projective}; in other words, they are the first integrals that define a solution of the HJE~\eqref{eq:HJE}. 
\end{theorem}

Now, the condition~\eqref{eq:poisson} is converted into the set of algebraic equations~\eqref{eq:infCond} for $n$ vectors $\bar{q}_1, \dots, \bar{q}_n \in \bm{R}^d$. 
If $\bar{q}_2, \dots, \bar{q}_n$ satisfy~\eqref{eq:infCond} for $k, l \in \{1,\dots,n\}$ and~\eqref{eq:projCond} with $\bar{q}_1$, the first integrals $f_2,\dots,f_n$ that define a solution of the HJE are obtained as $f_k = \phi_{\mathcal{B}}^{-1} \circ \psi_{\bar{z}}^{-1}(\bar{q}_k)\ (k=2,\dots,n)$. 
However, \eqref{eq:infCond} consists of an infinite number of algebraic equations, which cannot be solved or verified for some candidates of $\bar{q}_2, \dots, \bar{q}_n$ in finite time. 
The following section shows that the number of equations can be reduced to be finite.  

\subsection{Reduction to Finite Set of Algebraic Equations}
From the closure property of holonomic functions in Lemma~\ref{lem:closure}, the Poisson bracket $\{f_k, f_l\}$ is a holonomic function of $z$ if $f_k$ and $f_l$ are holonomic. 
Hence, a zero-dimensional ideal annihilating $\{f_k, f_l\}$ and the corresponding Pfaffian system exist. 
Furthermore, it can be shown that there exists a Pfaffian system satisfied by all Poisson brackets $\{f, g\}$ of any two functions $f, g \in \mathbb{S}$. 
\begin{lemma} \label{lem:matFiniteness}
	A finite set of multi-indices $\Gamma = \{\bm{0}, \gamma_1,\dots,\gamma_{t-1}\}$ exists
	such that for any pair of solutions $f, g \in \mathbb{S}$, a Pfaffian system
	\begin{equation}
		\partial_i r(z) = T_i(z) r(z) \quad (i=1,\dots,2n)  \label{eq:PoiPfaff}
	\end{equation}
	is satisfied by the vector-valued function $r = \mathcal{C} \bullet \{f, g\}$ with a vector of differential operators
	\begin{equation}
		\mathcal{C} = [1\ \partial^{\gamma_1}\ \cdots\ \partial^{\gamma_{t-1}}]^\T. \label{eq:PoiB}
	\end{equation}
\end{lemma}
\begin{proof}
	First, the Pfaffian system~\eqref{eq:PoiPfaff} is derived from the infinite set of matrices $\mathcal{D}^\alpha\Omega\ (\alpha \in \bm{Z}_{\ge 0}^{2n})$. 
	For the $\bm{R}(z)$-vector space $\mathrm{Skew}_d(\bm{R}(z))$, let $\mathbb{M} \subseteq \mathrm{Skew}_d(\bm{R}(z))$ be the subspace spanned by all matrices $\mathcal{D}^\alpha\Omega$, i.e., 
	\begin{equation}
		\mathbb{M} \coloneqq \mathrm{Span}_{\bm{R}(z)}\{\mathcal{D}^\alpha\Omega \mid \alpha \in \bm{Z}_{\ge 0}^{2n}\}. 
	\end{equation}
	The subspace $\mathbb{M}$ must be finite-dimensional because the dimension of $\mathrm{Skew}_d(\bm{R}(z))$ is $( d(d-1)/2 )$. 
	Let $t < ( d(d-1)/2 )$ be the dimension of $\mathbb{M}$ and the set of $t$ matrices $\{\mathcal{D}^{\gamma_0}\Omega, \dots, \mathcal{D}^{\gamma_{t-1}}\Omega\}$ be a basis of $\mathbb{M}$. 
	It should be noted that without loss of generality, we can select $\Gamma = \{\gamma_0, \gamma_1,\dots, \gamma_{t-1}\}$ such that $\gamma_0 = \bm{0} \in \bm{Z}_{\ge 0}^{2n}$ and $\gamma_0 \prec \gamma_1 \prec \cdots \prec \gamma_{t-1} \prec \alpha$ for all $\alpha \in \bm{Z}_{\ge 0}^{2n} \setminus \Gamma$ with a monomial order $\prec$ on $\bm{Z}_{\ge 0}^{2n}$. 

	For any $k \in \{0,\dots,t-1\}$ and $i \in \{1,\dots,2n\}$, there exist rational functions $c_{ikl} \in \bm{R}(z)\ (l=0,\dots,t-1)$ such that the following holds:
	\begin{equation}
		\mathcal{D}_{i}\mathcal{D}^{\gamma_k}\Omega = \sum_{l=0}^{t-1} c_{ikl}\mathcal{D}^{\gamma_l}\Omega. \label{eq:matPfaff}
	\end{equation}
	For every pair of functions $f, g \in \mathbb{S}$,~\eqref{eq:matPfaff} yields
	\begin{align}
		\partial_i\partial^{\gamma_k} \bullet \{f, g\} &= q_f^\T \{\mathcal{D}_i\mathcal{D}^{\gamma_k} \Omega \} q_g \notag \\
		&= \sum_{l=0}^{t-1} c_{ikl} \left[q_f^\T \{\mathcal{D}^{\gamma_l} \Omega\} q_g\right] \notag \\
		&= \sum_{l=0}^{t-1} c_{ikl} \left[ \partial^{\gamma_l} \bullet \{f, g\} \right], \label{eq:mat2poi}
	\end{align}
	where $q_f = \phi_{\mathcal{B}}(f)$ and $q_g = \phi_{\mathcal{B}}(g)$. 
	By combining~\eqref{eq:mat2poi} for $k = 0,\dots,t-1$ into one, we obtain the Pfaffian system~\eqref{eq:PoiPfaff} with matrices $T_i\ (i=1,\dots,2n)$ whose $(k, l)$-component is $c_{ikl}$ and vector $\mathcal{C}$ defined by~\eqref{eq:PoiB}. 

	Next, the proof of matrices $T_i\ (i=1,\dots,2n)$ satisfying the integrability condition~\eqref{eq:intCond} is presented. 
	Let us consider~\eqref{eq:matPfaff}. 
	By applying $\mathcal{D}_j$ with $j \neq i$ to both sides, we obtain
	\begin{align}
		\mathcal{D}_j\mathcal{D}_i\mathcal{D}^{\gamma_k}\Omega =& \mathcal{D}_{j} \left\{ \sum_{l=0}^{t-1} c_{ikl}\mathcal{D}^{\gamma_l}\Omega\right\} \notag\\
		=& \sum_{l=0}^{t-1} \left\{ c_{ikl}\mathcal{D}_j\mathcal{D}^{\gamma_l}\Omega + (\partial_j c_{ikl})\mathcal{D}^{\gamma_l}\Omega \right\} \notag\\
		=& \sum_{l=0}^{t-1} \left\{ c_{ikl}\sum_{s=0}^{t-1} c_{jls}\mathcal{D}^{\gamma_s}\Omega + (\partial_j c_{ikl})\mathcal{D}^{\gamma_l}\Omega \right\} \notag \\
		=& \sum_{l=0}^{t-1} \left\{\sum_{s=0}^{t-1} c_{iks}c_{jsl} + (\partial_j c_{ikl}) \right\}\mathcal{D}^{\gamma_l}\Omega. \label{eq:deriv_ji}
	\end{align}
	By replacing the subscript $i$ with $j$, we have 
	\begin{equation}
		\mathcal{D}_i\mathcal{D}_j\mathcal{D}^{\gamma_k}\Omega = \sum_{l=0}^{t-1} \left\{\sum_{s=0}^{t-1} c_{jks}c_{isl} + (\partial_i c_{jkl}) \right\}\mathcal{D}^{\gamma_l}\Omega. \label{eq:deriv_ij}
	\end{equation}
	The right-hand sides of~\eqref{eq:deriv_ji} and~\eqref{eq:deriv_ij} are equal to each other because $\mathcal{D}_i\mathcal{D}_j\mathcal{D}^{\gamma_k}\Omega = \mathcal{D}_j\mathcal{D}_i\mathcal{D}^{\gamma_k}\Omega$ holds by considering~\eqref{eq:intCond} and~\eqref{eq:defD}. 
	In particular, the coefficients of each matrix are equal to each other because matrices $\mathcal{D}^{\gamma_0}\Omega,\dots,\mathcal{D}^{\gamma_{t-1}}\Omega$ are linearly independent, thereby yielding the integrability condition for $T_1,\dots,T_{2n}$: 
	\begin{equation}
		\partial_jT_i + T_iT_j = \partial_iT_j + T_jT_i \quad (1 \le i < j \le 2n). 
	\end{equation}
	This completes the proof. 
\end{proof}

Thus, from Lemma~\ref{lem:matFiniteness}, the main result is obtained as follows. 
\begin{theorem} \label{thm:finiteness}
	With the same notations as Lemma~\ref{lem:matFiniteness}, let $E(z) \in \bm{R}[z]$ be the least common multiple of the denominators included in $T_1,\dots,T_{2n}$. 
	For any fixed point $\bar{z} \in U' \coloneqq \{z \in U \mid E(z) \neq 0\}$, constant vectors $\bar{q}_k \in \bm{R}^d\ (k=1,\dots,n)$ satisfy~\eqref{eq:infCond}
	if and only if they satisfy a finite number of algebraic equations:
	\begin{equation}
		\bar{q}_k^\T \mathcal{D}^{\gamma} \Omega(\bar{z}) \bar{q}_l = 0 \quad 
		(\gamma \in \Gamma,\ 1 \leq k < l \leq n). \label{eq:finCond}
	\end{equation}
\end{theorem}
\begin{proof}
	From Lemma~\ref{lem:finiteness} and the Pfaffian system~\eqref{eq:PoiPfaff}, the Poisson bracket $\{f_k, f_l\}(z)$ is constantly equal to $0$ on $U'$ if and only if
	$(\mathcal{C} \bullet \{f_k, f_l\})(\bar{z}) = \bm{0} \in \bm{R}^t$ holds at any one point $\bar{z} \in U'$. 
	This implies that~\eqref{eq:infCond} holds if and only if~\eqref{eq:finCond} is valid for $\bar{q}_k = \psi_{\bar{z}} \circ \phi_{\mathcal{B}}(F_k)\ (k=1,\dots,n)$. 
\end{proof}

Theorem~\ref{thm:finiteness} ensures that we can determine $f_2, \dots, f_n \in \mathbb{S}$ that define a solution of the HJE with $h$ by finding $\bar{q}_2, \dots, \bar{q}_n \in \bm{R}^d$ that satisfy~\eqref{eq:finCond} and~\eqref{eq:projCond} with $\bar{q}_1 = \psi_{\bar{z}} \circ \phi_{\mathcal{B}}(h)$.  
More precisely, $f_k$ is specified as the analytic function that satisfies $q_k = \mathcal{B} \bullet f_k$ for the solution $q_k(z)$ of the Pfaffian system~\eqref{eq:Pfaff} defined by $A_1,\dots,A_{2n}$ with boundary condition $q_k(\bar{z}) = \bar{q}_k$. 
Although it is difficult to analytically solve the Pfaffian system, we can numerically evaluate $f_k$ and its derivatives using the holonomic gradient method (HGM)~\cite{Nakayama2011}. 
Solving~\eqref{eq:finCond} is easier than directly finding the first integrals that satisfy conditions~\eqref{eq:poisson} and~\eqref{eq:projective} by solving the Lagrange-Charpit system~\cite{Sakamoto2002}. 

It should be noted that for the solutions of algebraic equations~\eqref{eq:finCond} to exist, the dimension $d$ of $\bar{q}_k$ must be sufficiently larger than the number of equations $t$. 
There are more than one zero-dimensional ideals that annihilate the holonomic Hamiltonian, and the dimensions $d$ and $t$ depend on the choice of the ideal. 
It is a part of future work to clarify which zero-dimensional ideal yields $d$ sufficiently larger than $t$. 

\section{Derivation of Finite Equation Set}
To find the algebraic equations~\eqref{eq:finCond}, an appropriate finite set of multi-indices $\Gamma = \{\bm{0}, \gamma_1, \dots, \gamma_{t-1}\} \subset \bm{Z}_{\geq 0}^{2n}$ must be determined. 
In the proof of Lemma~\ref{lem:matFiniteness}, the multi-indices yield a basis $\Omega, \mathcal{D}^{\gamma_1}\Omega, \dots, \mathcal{D}^{\gamma_{t-1}}\Omega$ of the $t$-dimensional $\bm{R}(z)$-vector space $\mathbb{M}$. 
It should be noted that finitely many matrices in $\{\mathcal{D}^{\alpha}\Omega \mid \alpha \in \bm{Z}_{\ge 0}^{2n}\}$ can be recursively computed by applying $\mathcal{D}_i\ (i=1,\dots,2n)$ to $\Omega$. 
Moreover, for the finite set of matrices, the computation of its linearly independent subset is a straightforward process. 
However, we cannot determine whether the linearly independent set is sufficient to span $\mathbb{M}$ because dimension $t$ is unknown. 
In this section, a method is proposed to compute $\Gamma$, which provides a basis of $\mathbb{M}$, from a finite number of matrices $\mathcal{D}^{\alpha}\Omega$. 

For a finite number of matrices $\mathcal{D}^{\alpha_1} \Omega,\dots,\allowbreak\mathcal{D}^{\alpha_s}\Omega$, suppose that there exist several $\bm{R}(z)$-linear relations between them: 
\begin{multline}
	c_{j, \alpha_1}\mathcal{D}^{\alpha_1}\Omega + \cdots + c_{j, \alpha_s}\mathcal{D}^{\alpha_s}\Omega = 0 \\
	(c_{j, \alpha_i} \in \bm{R}(z);\ i=1,\dots,s;\ j=1,\dots,r). \label{eq:linEq}
\end{multline}
Equations~\eqref{eq:linEq} with~\eqref{eq:pd2D} indicate that any Poisson bracket $\{f, g\}$ with $f, g \in \mathbb{S}(\mathcal{P}_1,\dots,\mathcal{P}_{2n})$ is annihilated by differential operators
$\mathcal{Q}_{j} \coloneqq c_{j, \alpha_1}\partial^{\alpha_1} + \cdots + c_{j, \alpha_s}\partial^{\alpha_s}\ (j=1,\dots,r)$. 
If the ideal $I = \langle \mathcal{Q}_1,\dots,\mathcal{Q}_r \rangle$ is zero-dimensional, the quotient space $\bm{R}(z)\langle \partial \rangle / I$ is a finite dimensional $\bm{R}(z)$-vector space with a basis $\{1, \partial^{\beta_1}, \dots, \partial^{\beta_{t'-1}}\}$ comprising the standard monomials~\cite{Saito2000} of $I$.  
Hence, any $\partial^\alpha\ (\alpha \in \bm{Z}_{\ge 0}^{2n})$ can be expressed as a unique $\bm{R}(z)$-linear combination of the elements of the basis. 
This implies that the set of matrices $B \coloneqq \{\Omega, \mathcal{D}^{\beta_1}\Omega, \dots, \mathcal{D}^{\beta_{t'-1}}\Omega\}$ spans $\mathbb{M}$. 
If $B$ is linearly independent, dimension $t$ is given as $t = t'$ and $\Gamma = \{\bm{0}, \beta_1,\dots,\beta_{t'-1}\}$. 
If not, the minimal linearly independent subset of $B$ is a basis of $\mathbb{M}$. 
The procedure for computing $\Gamma$ is summarized in Algorithm~\ref{alg:finite}. 
\begin{algorithm}[t]
	\caption{Derivation of $\Gamma$ in~\eqref{eq:finCond} \label{alg:finite}}
	\begin{algorithmic}[1]
		\Require{Matrices of rational functions $\Omega, A_1, A_2, \dots, A_{2n} \in \bm{R}(z)^{d \times d}$}
		\Ensure{Finite set of multi-indices $\Gamma$}
		\State{Compute finite number of matrices $\mathcal{D}^\alpha \Omega$}\label{algline:getMat}
		\State{Find $\bm{R}(z)$-linear relations between them and corresponding differential operators $\mathcal{Q}_j\ (j=1,\dots,r)$}
		\If{$I = \langle \mathcal{Q}_1,\dots,\mathcal{Q}_r\rangle$ is not zero-dimensional}
			\State{Get back to line~\ref{algline:getMat}}
		\EndIf
		\State{Find a basis $\{1, \partial^{\beta_1}, \dots, \partial^{\beta_{t'-1}}\}$ of quotient space $\bm{R}(z) / I$}
		\If{$B = \{\Omega, \mathcal{D}^{\beta_{1}}\Omega, \dots, \mathcal{D}^{\beta_{t'-1}}\Omega\}$ is linearly independent}
			\State\Return{$\Gamma \gets \{1, \beta_1, \dots, \beta^{t'-1}\}$}
		\Else
			\State{Find minimal linearly independent subset $B' \subset B$}
			\State\Return{$\Gamma \gets \{\mbox{Multi-indices of matrices in }B'\}$}
		\EndIf
	\end{algorithmic}
\end{algorithm}

\section{Numerical Example}
Consider the HJE for the following Hamiltonian. 
\begin{equation}
	h(x, p) = -2p_1\sin(x_1) + 2x_2p_2 - ap_2^2 + bx_1^4, \label{eq:exHam}
\end{equation}
where $x = [x_1\ x_2]^\T$, $p = [p_1\ p_2]^\T$, and $a, b \in \bm{R}$. 
To apply the proposed method to this HJE, $h(x, p)$ should be verified to be holonomic. 
By virtue of Lemma~\ref{lem:closure}, it is sufficient to show that each term of $h(x, p)$ is a holonomic function. 
For example, after several differentiations, it was found that the first term $-2p_1\sin(x_1)$ is annihilated by the following four differential operators of the form~\eqref{eq:defHF}: $\mathcal{P}_1 = \partial_{x_1}^2 + 1$, $\mathcal{P}_2 = \partial_{x_2}$, $\mathcal{P}_3 = \partial_{p_1}$, and $\mathcal{P}_4 = p_2\partial_{p_2} - 1$. 
For the other terms, the appropriate differential operators were found by differentiating them several times, which guarantees that every term is holonomic and $h(x, p)$ is also holonomic from Lemma~\ref{lem:closure}. 
A zero-dimensional ideal that annihilates $h$ was determined by computing the intersection of the ideals for all terms (see~\cite{Oaku2003} for details). 
It should be noted that this computation of intersection can be performed independently of parameters $a$ and $b$. 

Using the symbolic computation of differential operators, a Pfaffian system satisfied by $h(x, p)$, i.e., the vector of differential operators $\mathcal{B}$ and matrices of rational functions $A_i(z) \in \bm{R}(z)\ (i=1,\dots,4)$ are obtained as 
\begin{align*}
	\mathcal{B} &= [1\ \partial_{p_2}\ \partial_{p_1}\ \partial_{x_2}\ \partial_{x_1}]^\T, \\
	A_{x_1} &= \begin{bmatrix}
		0 & 0 & 0 & 0 & 1 \\
		0 & 0 & 0 & 0 & 0 \\
		-\frac{4}{x_1p_1} & \frac{2p_2}{x_1p_1} & \frac{4}{x_1} & \frac{2x_2}{x_1p_1} & \frac{1}{p_1} \\
		0 & 0 & 0 & 0 & 0 \\
		\frac{12}{x_1^2} & -\frac{6p_2}{x_1^2} & \frac{-12p_1}{x_1^2} - p_1 & -\frac{6x_2}{x_1^2} & 0
	\end{bmatrix}, 
\end{align*}
where $A_2$, $A_3$, and $A_4$ are omitted owing to space limitations. 
From Lemma~\ref{lem:finiteness}, the analytic solutions of the Pfaffian system defined on a simply connected domain $U$ in $\{z \in \bm{R}^4 \mid x_1p_1p_2 \neq 0\}$ constitute the $5$-dimensional real vector space. 

The matrices $B_x, B_p \in \bm{R}(z)^{2 \times 5}$ and $\Omega \in \bm{R}(z)^{5 \times 5}$ are derived from $A_1,\dots,A_4$ as constant matrices:
\begin{multline*}
	B_x = \begin{bmatrix}
		0 & 0 & 0 & 0 & 1 \\
		0 & 0 & 0 & 1 & 0 
	\end{bmatrix},  
	B_p = \begin{bmatrix}
		0 & 0 & 1 & 0 & 0 \\
		0 & 1 & 0 & 0 & 0 
	\end{bmatrix}, 
\end{multline*}
and $\Omega = B_p^\T b_x - B_x^\T B_p \in \bm{R}(z)^{5 \times 5}$. 
Using Algorithm~\ref{alg:finite}, $\Gamma = \{(0,0,0,0), (1,0,0,0), (0,0,1,0)\}$ is obtained, which indicates from Theorem~\ref{thm:finiteness} that~\eqref{eq:infCond} is valid if and only if 
\begin{equation}
	\bar{q}_1^\T \Omega(\bar{z}) \bar{q}_2 = \bar{q}_1^\T \mathcal{D}_1 \Omega(\bar{z}) \bar{q}_2 = \bar{q}_1^\T \mathcal{D}_3\Omega(\bar{z}) \bar{q}_2 = 0 \label{eq:exFinCond}
\end{equation}
holds. 
In this example, the condition~\eqref{eq:projCond} was obtained as
\begin{equation}
	\det \left\{
	 B_p [\bar{q}_1\ \bar{q}_2] \right\} = \bar{q}_{1,3}\bar{q}_{2,2} - \bar{q}_{2,3}\bar{q}_{1,2} \neq 0. \label{eq:exProjCond}
\end{equation}
Consequently, any two vectors satisfying~\eqref{eq:exFinCond} and~\eqref{eq:exProjCond} yield the first integrals of the Hamiltonian system associated with~\eqref{eq:exHam} that define a solution of the HJE. 
It should be noted that all aforementioned computations can be performed symbolically and include no approximations. 

Based on the denominators of the components of $A_1,\dots,A_4$, let $U = \{z \in \bm{R}^4 \mid x_1 > 0, p_1 > 0, p_2 > 0\}$, and let $\mathbb{VS}(A_1,\dots,A_4)$ and $\mathbb{S}(\mathcal{P}_1, \dots, \mathcal{P}_4)$ be the solution spaces defined on it. 
For a fixed point $\bar{z} = [\pi/6\ 1\ b(\pi/6)^4\ 2/a]^\T$, which satisfies $h(\bar{z}) = 0$, we seek the other first integral $f_2$ that defines a solution of the HJE. 
The constant vector $\bar{q}_1$ is computed as 
\[
	\bar{q}_1 = (\mathcal{B} \bullet h)(\bar{z}) 
	= \left[ 0\ -2\ -1\ \frac{4}{a}\ -\frac{b\pi^3(\sqrt{3}\pi - 24)}{1296}\right]^\T \in \bm{R}^5
\]
For any constant vector $\bar{q}_2$ that satisfies~\eqref{eq:exFinCond} and~\eqref{eq:projCond} with $\bar{q}_1$, the other first integral that defines a solution of the HJE is obtained as $f_2 = \phi_{\mathcal{B}}^{-1} \circ \psi_{\bar{z}}^{-1}(\bar{q}_2)$. 
An appropriate $\bar{q}_2 \in \bm{R}^5$ satisfying~\eqref{eq:exFinCond} and~\eqref{eq:exProjCond} can be selected because the number of equations is $3$. 
For example, $\bar{q}_2 = [0\ a\ 0\ -2\ 0]^\T$ satisfies both conditions. 
Although it is difficult to compute its explicit expression, $f_2$ is uniquely characterized as the first component of the solution of the Pfaffian system defined by $A_1,\dots,A_4$ with boundary condition $Q(\bar{z}) = \bar{q}_2$ and can be evaluated using the HGM~\cite{Nakayama2011}. 

\section{Conclusion}
In this paper, the solutions of the HJE with holonomic Hamiltonian were investigated. 
A solution of the HJE is characterized by a finite number of the first integrals of the corresponding Hamiltonian system. 
Under the assumption that the Hamiltonian is holonomic, the finite dimensional solution space of the Pfaffian system associated with the holonomic Hamiltonian can be used to obtain the condition for the other first integrals to characterize a solution of the HJE. 
We obtained the condition as a finite number of algebraic equations, which can be computed through the symbolic computation of differential operators. 
Although the algebraic equations do not necessarily have a solution, they can be easily solved if at least one solution exists. 

Future work directions of this study include the investigations related to the range of problems or conditions on the Hamiltonian to guarantee the existence of solutions of the algebraic equations. 
When the HJE is considered in control theory, a particular solution of the HJE called the stabilizing solution is required. 
The characterization of the stabilizing solution in terms of holonomic functions, zero-dimensional ideals, and Pfaffian systems will also be considered in future work.

\end{document}